\documentclass[12pt]{article}

\usepackage{fullpage}
\usepackage{amsmath, amsthm, amssymb}
\usepackage{verbatim}
\usepackage[T1]{fontenc}
\usepackage[tt=false,type1=true]{libertine}
\usepackage[varqu]{zi4}
\usepackage[libertine]{newtxmath}
\usepackage[sort]{natbib}
\usepackage{xcolor}
\usepackage[pagebackref=true,colorlinks]{hyperref}
\hypersetup{linkcolor=blue,filecolor=blue,citecolor=blue,urlcolor=blue}
\usepackage{graphicx}
\usepackage{authblk}

\newtheorem{lemma}{Lemma}
\newtheorem{theorem}{Theorem}

\theoremstyle{definition}

\newcommand{\btheta}{{\boldsymbol{\theta}}}
\newcommand{\poa}{\mathsf{PoA}}
\newcommand{\argmax}{\operatorname*{argmax}}
\newcommand{\rw}{\mathsf{rw}}
\newcommand{\agg}{\mathrm{agg}}
\newcommand{\con}{\mathrm{con}}

\title{Autobidding Auctions in the Presence of User Costs}
\date{}

\author[1]{Yuan Deng}
\author[1]{Jieming Mao}
\author[1]{Vahab Mirrokni}
\author[2]{Hanrui Zhang}
\author[1]{Song Zuo}
\affil[1]{Google, \texttt{\{dengyuan,maojm,mirrokni,szuo\}@google.com}}
\affil[2]{Carnegie Mellon University, \texttt{hanruiz1@cs.cmu.edu}}

\begin{document}

\maketitle

\begin{abstract}
We study autobidding ad auctions with user costs, where each bidder is value-maximizing subject to a return-over-investment (ROI) constraint, and the seller aims to maximize the social welfare taking into consideration the user's cost of viewing an ad. We show that in the worst case, the approximation ratio of social welfare by running the vanilla VCG auctions with user costs could as bad as $0$. To improve the performance of VCG, We propose a new variant of VCG based on properly chosen cost multipliers, and prove that there exist auction-dependent and bidder-dependent cost multipliers that guarantee approximation ratios of $1/2$ and $1/4$ respectively in terms of the social welfare. 
\end{abstract}

\section{Introduction}

As 2022, over 90\% of online advertising transact through automation technologies~\citep{programmatic22}. As one key application of automation in online advertising, auto-bidding significantly simplifies the tasks of advertisers by delegating the bidding to automated agents who submit real-time bids on behalf of the advertisers in each ad auction to optimize with pre-specified high-level objectives and constraints.
One popular and representative example of auto-bidding strategies is to maximize the total advertiser value subject to a target return-on-investment (ROI) constraint declared by the advertiser. This type of auto-bidders is referred as {\em value maximizers} by a recent line of research~\citep{deng2021towards,balseiro2021landscape} as these auto-bidders aim to maximize the total advertiser value from ad auctions, such as the (weighted) number of clicks or conversions, while guaranteeing the ROI no less than the advertiser-specified threshold.

The differences in behavior models between the classic (quasi-linear) {\em utility maximizers} (who optimize for value minus payment) and value maximizers results in different conclusions for many basic problems in auction theory. As one demonstrating example, Vickrey-Clarke-Groves (VCG) auction is known to achieve the optimal social welfare when bidders are utility maximizers. However, \citet{aggarwal2019autobidding} showed that in the worst case, the welfare of VCG auction can be as bad as $1/2$ of the optimal social welfare when bidders are value maximizers. Similar distinctions have been discovered by recent works revisiting the effectiveness of various existing mechanisms in the presence of value maximizers \citep{balseiro2021landscape,Mehta22,rFPApaper,deng2022efficiency}.

\paragraph{User Costs}
User costs to the auctioneer do not receive much attention in the literature on auction design for online advertising with utility maximizers, despite of the prevalence of the externalities on user experience of displaying ads to users~\citep{AbramsS07}. One reason is that when bidders are utility maximizers, many commonly studied auctions and their properties directly generalize from models without user costs to models with user costs. For example, VCG auction by definition directly applies to the model with user costs and continues to achieve the optimal social welfare. \citet{AbramsS07} generalize General-Second-Price (GSP) auction to accommodate the cost from user experience impact in online advertising by treating the original bids minus the costs as the effective bids from advertisers. They further prove that GSP preserves the property that there exists an equilibrium achieving the optimal social welfare after the generalization. 

\paragraph{Challenges from User Cost in Autobidding.} However, such nice generalization of concepts and results for user costs no longer works when the bidders are value maximizers. As we will show later in Section~\ref{sec:global}, VCG auction can have $0$ approximation to the optimal social welfare in the worst case when bidders are value maximizers. One intuitive explanation of the worsened approximation guarantee of VCG auction in this case lies behind the so-called {\em uniform-bidding} strategy widely adopted by auto-bidding agents~\citep{feldman2007budget,bateni2014multiplicative}. Under uniform-bidding strategies, each bidder bids its value of each auction times the {\em bid multiplier} that is constant across all auctions for this bidder. From the perspective of value maximizers, the best response is always achievable by uniform-bidding with a properly optimized bid multiplier under VCG auction with/without user costs \citep{aggarwal2019autobidding,deng2021towards}. 
In contrast, from the perspective of the auctioneer, bidders participating in the same auction could have different bid multipliers, meaning that their base values are scaled up by different magnitudes. Thus, the ranking of their bids can be very different from the efficient ranking, leading to an inefficient allocation. The inefficiency becomes worse off in the presence of user costs, since the cost associated to each bidder is not scaled by the corresponding bid multiplier in VCG auction. For bidders with bid multipliers much larger than $1$, the effect of their costs in ranking is then vanished, resulting in additional allocation inefficiency.

Based on the above observation, one na\"{i}ve and minimal fix to VCG auction is to artificially scale user costs by a multiplier and hope that VCG auction can retain a good approximation guarantee with a properly chosen {\em cost multiplier}. Unfortunately, using one global cost multiplier is not enough as we will show later in Section~\ref{sec:global}, the approximation guarantee for VCG auction with a global cost multiplier is still $0$ in the worst case.

\paragraph{Our Results} In this paper, we present how the constant approximation of social welfare can be retained with one step further on top of the above idea, namely, VCG auction with cost multipliers. In particular, we prove that there exists a set of cost multipliers corresponding to each auction (i.e., auction-dependent but bidder-independent cost multipliers) such that the welfare in any equilibrium among value maximizers under VCG auction with those cost multipliers is at least $1/2$ of the optimal social welfare (Section~\ref{sec:auction-dependent}). In addition, we further show that there exists a set of cost multipliers corresponding to each bidder (i.e., bidder-dependent but auction-independent cost multipliers) such that the welfare performance is at least $1/4$ of the optimal social welfare (Section~\ref{sec:bidder-dependent}).


\subsection{Related Work}

Since the seminal work of Vickrey–Clarke–Groves (VCG) auctions~\citep{vickrey1961counterspeculation,clarke1971multipart,groves1973incentives} and Myerson's auction~\citep{myerson1981optimal}, auction theory has found numerous practical applications. Examples includes combinatorial auctions for reallocating radio frequencies~\citep{cramton2006combinatorial}, generalized second-price auctions (GSP) and dynamic auctions for online advertising~\citep{aggarwal2006truthful,edelman2007internet,varian2007position,mirrokni2018dynamic,mirrokni2020non,balseiro2022dynamic}.
Most of the work in auction theory assume bidders are utility maximizers.

More recently, there's a line of research focusing on value maximizing bidders motivated by auto-bidding (e.g. target CPA bidders and target ROAS bidders) in online advertising. \citet{aggarwal2019autobidding} characterizes the optimal bidding strategies in truthful auction for value maximizing bidders and shows that in the worst case, the welfare of VCG auction can be as bad as $1/2$ of the optimal social welfare. \citet{deng2021towards,BalseiroDMMZ21, deng2022efficiency} show how boosts and reserves can be used to improve the welfare efficiency in auctions with value maximizing bidders. \citet{balseiro2021landscape, BalseiroDMMZ22} characterize the revenue-optimal single-stage auctions with value maximizers with or without budget constraints under various information structure. In addition to the above, there are other related work studying slightly different bidding models, e.g.  \citet{golrezaei2021auctions, goel2014clinching,goel2019pareto,babaioff2020non} study utility maximizers with ROI constraints, and \citet{fadaei2016truthfulness,wilkens2016mechanism,wilkens2017gsp} study value maximizing bidders with per-auction target constraints.

As far as we know, \citet{AbramsS07} is the first paper to study auction design with user costs for utility maximizing bidders, and \citet{LiLL13} extends the study to include position-specific information. User costs are motivated by studies on ``banner blindness'' \cite{Benway1998BannerBW, Cho04whydo, doi:10.1002/dir.10063} which show how user experience affects future user behavior with ads.

\section{Preliminaries}

\paragraph{Ad auctions with user costs.}
Following prior work on autobidding ad auctions~\citep{aggarwal2019autobidding,deng2021towards,BalseiroDMMZ21,Mehta22}, we consider the following multi-auction model.
There are $n$ bidders $[n] = \{1, 2, \dots, n\}$ and $m$ auctions $[m] = \{1, 2, \dots, m\}$.
We generally use $i$ to index bidders, and $j$ to index auctions.
In each auction $j$, each bidder $i$ has a {\em value} $v_{i, j}$, as well as a {\em user cost} $c_{i, j}$.
The user cost is suffered by the user who views the ad provided by bidder $i$ in auction $j$ when $i$ wins.
As such, it does not affect the behavior of the bidder directly, but counts negatively towards the social welfare.
We assume the value $v_{i, j}$ is private, but the user cost $c_{i, j}$ is public and can be directly observed by the auction mechanism.
In each auction $j$, each bidder $i$ submits a {\em bid} $b_{i, j}$, and all bidders' bids $\{b_{i', j}\}_{i'}$ and user costs $\{c_{i', j}\}_{i'}$ together determine each $i$'s {\em allocation} $x_{i, j}$ (i.e., $x_{i, j} = 1$ if $i$ wins in $j$, and $0$ otherwise) and {\em payment} $p_{i, j}$, in the way prescribed by the auction mechanism (we will discuss auction mechanisms momentarily).
We omit the dependence of $x_{i, j}$ and $p_{i, j}$ on $\{b_{i', j}\}_{i'}$ when it is clear from the context.

\paragraph{ROI-constrained value maximizers.}
We assume bidders are ROI-constrained value maximizers.
That is, each bidder $i$ maximizes the total value $i$ receives in all auctions, subject to the constraint that the ratio between this total value and the total payment $i$ makes in all auctions is at least some target quantity.
Following prior work~\citep{deng2021towards,BalseiroDMMZ21,deng2022efficiency}, without loss of generality, we assume this target ratio is $1$ for all bidders.
Formally, each bidder $i$ solves the following optimization problem to decide how to bid in all auctions:
\begin{align*}
    \max & \quad \sum_j x_{i, j} \cdot v_{i, j} \\
    \text{subject to} & \quad \sum_j (x_{i, j} \cdot v_{i, j} - p_{i, j}) \ge 0.
\end{align*}
We always assume there is a way of bidding such that the ROI constraint is satisfied (this is true for all quasilinear truthful auction mechanisms).

\paragraph{Quasilinear truthfulness and uniform bidding.}
In this paper, we focus on auction mechanisms that are quasilinear truthful.
An auction mechanism is quasilinear truthful, if each bidder would be motivated to bid their true value if they were to maximize their quasilinear utility.
In other words, for each $i$ and $j$, fixing any $v_{i, j}$ and $\{b_{i', j}\}_{i' \ne i}$, the following holds for any possible bid $b$:
\begin{align*}
    & x_{i, j}(b_{i, j} = v_{i, j}, \{b_{i', j}\}_{i' \ne i}) \cdot v_{i, j} - p_{i, j}(b_{i, j} = v_{i, j}, \{b_{i', j}\}_{i' \ne i}) \\
    \ge\ & x_{i, j}(b_{i, j} = b, \{b_{i', j}\}_{i' \ne i}) \cdot v_{i, j} - p_{i, j}(b_{i, j} = b, \{b_{i', j}\}_{i' \ne i}).
\end{align*}
One example of quasilinear truthful mechanisms is the {\em second-price auction}: in each auction $j$, the bidder $i$ with the highest bid wins, and the payment $i$ makes is the second highest bid.
In addition, quasilinear truthful mechanisms enjoy the following nice property: without loss of generality, the optimal bidding strategy of an ROI-constrained value maximizer is always {\em uniform bidding}~\citep{aggarwal2019autobidding}.
That is, for each bidder $i$, there exists a bid multiplier $\theta_i \ge 1$, such that fixing other bidders' bids, the optimal (i.e., value-maximizing subject to the ROI constraint) bidding strategy of $i$ is to bid $b_{i, j} = \theta_i \cdot v_{i, j}$ in each auction $j$.
Note that the bid multiplier is never smaller than $1$, which in particular means no bidder should ever bid below their true value.
In the rest of the paper, we only consider such bidding strategies, and sometimes use the corresponding bid multiplier to represent a bidding strategy.

\paragraph{Equilibria and the price of anarchy.}
To measure the efficiency of an auction mechanism, we consider the behavior of bidders in equilibrium under that mechanism.
We say the bidding strategies $\{\theta_i\}_i$ form an equilibrium, if each bidder $i$'s bidding strategy $\theta_i$ is a best response to all other bidders' strategies $\btheta_{-i} = \{\theta_{i'}\}_{i' \ne i}$. Moreover, we use the price of anarchy (PoA) to quantify the efficiency of a mechanism.
The PoA of a mechanism is the ratio between the worst-case welfare in equilibrium under that mechanism and the optimal welfare.
Formally, fixing a mechanism,
\[
    \poa = \inf_{\substack{n, m, \{v_{i, j}\}, \{c_{i, j}\} \\ \{\theta_i\}_i\text{ form an equilibrium}}} \frac{\sum_{i, j} x_{i, j} \cdot (v_{i, j} - c_{i, j})}{\sum_j \max\left\{0, \max_i (v_{i, j} - c_{i, j})\right\}}.
\]
There are $2$ things worth noting in the above definition:
\begin{itemize}
    \item If every bidder would contribute negatively to the social welfare if they were to win in a particular auction $j$ (i.e., $\max_i (v_{i, j} - c_{i, j}) < 0$), then auction $j$ contributes $0$ to the optimal welfare.
    \item In order for the PoA to be well defined, we always assume the optimal welfare is positive.
\end{itemize}

\section{Challenges Posed by User Costs}

To gain some intuition, we first briefly review existing approaches to autobidding auctions through quasilinear truthful mechanisms, particularly the second-price auction, and discuss why it fails in the presence of user costs.
In a second-price auction, the bidder with the highest bid wins, and the payment of the winner is the second highest bid.
It is known that the second-price auction has a PoA of $1/2$~\citep{aggarwal2019autobidding} when (1) bidders are ROI-constrained value maximizers and (2) there is no user cost.
Variants of the second-price auction has been analyzed in other cost-free settings, sometimes leading to improved efficiency guarantees (e.g., in settings with machine-learned advice~\citep{BalseiroDMMZ21}).
Nonetheless, in the core of these analyses lies the following conceptually simple argument:

\begin{theorem}[\citep{aggarwal2019autobidding,deng2021towards,BalseiroDMMZ21}] Suppose there is no user cost. The second-price auction achieves a PoA of $1/2$.
\end{theorem}
\begin{proof}
    Suppose there is no user cost.
    For each auction $j$, consider the ``rightful winner'' $\rw(j)$ who has the highest value in $j$, i.e., $\rw(j) = \argmax_i v_{i, j}$ (ties are broken arbitrarily).
    Since the second-price auction is quasilinear truthful, in equilibrium, all bidders, including the rightful winner, must bid at least their own value, i.e., $b_{i, j} \ge v_{i, j}$ for all $i$ and $j$.
    Therefore, in each auction $j$, one of the following two cases must happen: (1) the rightful winner $\rw(j)$ wins, and collects value $v_{i, j}$ from auction $j$, or (2) the rightful winner $\rw(j)$ does not win, in which case the second highest bid is at least $v_{\rw(j), j} = \max_i v_{i, j}$, and therefore the winner must pay at least this amount in auction $j$.

    Summing over all auctions, the above reasoning implies that the sum of the following two quantities must be at least the optimal welfare: (1) the total value in all auctions collected by the respective rightful winners, and (2) the total payment made in all auctions by bidders other than the respective rightful winners.
    The first quantity clearly lower bounds the welfare in equilibrium.
    In fact, the second quantity is also a lower bound of the welfare: recall that the ROI constraints require that the total value each bidder receives should be at least the total payment that bidder makes, and summing over all bidders, one can see the second quantity also lower bounds the welfare in equilibrium.
    Now since the sum of the two quantities is at least the optimal welfare, at least one of the two must be no smaller than half of the optimal welfare, which means the welfare in equilibrium is no smaller than half of the optimal welfare.
    This gives a lower bound of $1/2$ on the PoA of the second-price auction.
\end{proof}

In words, the above argument lower bounds the total value collected in equilibrium and the total payment made in equilibrium, respectively, and use the larger one between the two as a lower bound of the welfare in equilibrium.
Given this argument, it might appear that the most natural mechanism in the presence of user costs is the second price auction with ``cost-adjusted bids''.
That is, bidders are sorted by the following adjusted bid: $b_{i, j} - c_{i, j}$ (negative adjusted bids are discarded); the bidder with the highest adjusted bid wins, and pays the minimum bid to win (by Myerson's characterization~\citep{myerson1981optimal}, this is the unique payment rule that guarantees quasilinear truthfulness).\footnote{
    To make this concrete, suppose there are 3 bidders in some auction $j$.
    The bids and costs are: $b_{1, j} = 5$, $b_{2, j} = 3$, $b_{3, j} = 4$, $c_{1, j} = 1$, $c_{2, j} = 2$, and $c_{3, j} = 1$.
    Then the cost-adjusted bids are $b_{1, j} - c_{1, j} = 4$, $b_{2, j} - c_{2, j} = 1$, and $b_{3, j} - c_{3, j} = 3$, so bidder $1$ wins.
    Bidder $3$ has the second highest cost-adjusted bid, so the minimum bid that bidder $1$ has to make in order to win is $b_{3, j} - c{3, j} + c{1, j} = 4$, so bidder $1$ pays $p_{1, j} = 4$.
}

Now let us try to generalize the argument above to this new mechanism.
For each auction $j$, the rightful winner $\rw(j)$ is the one with the largest cost-adjusted value, i.e., $\rw(j) = \argmax_i (v_{i, j} - c_{i, j})$.
Again there are two cases: (1) if the rightful winner $\rw(j)$ wins in auction $j$, then the welfare collected by $\rw(j)$ is $v_{\rw(j), j} - c_{\rw(j), j}$; (2) if the rightful winner does not win, then the second highest cost-adjusted bid must be at least $\max_i (v_{i, j} - c_{i, j})$, so payment made by the winner (say $i^*$) is at least $c_{i^*, j} + \max_i (v_{i, j} - c_{i, j})$.

Now one might be tempted to apply the following (faulty) reasoning and argue the PoA is again at least $1/2$: in case (1), the contribution of auction $j$ to the total value collected minus user cost incurred is $v_{\rw(j), j} - c_{\rw(j), j} = \max_i (v_{i, j} - c_{i, j})$, which is the same as the contribution of $j$ to the optimal welfare.
Similarly, in case (2), the contribution of $j$ to the total payment made minus user cost incurred is $\max_i (v_{i, j} - c_{i, j})$, which is again the same as the contribution of $j$ to the optimal welfare.
Now if we sum over all $j$, then we know that the sum of (1) the total value minus user cost and (2) the total payment minus user cost is at least the optimal welfare, so the larger one is at least half of the optimal welfare, which gives a PoA lower bound of $1/2$.

The problem of this reasoning is in the last step: when summing over $j$, the user cost is only counted once, so what we actually get is: the sum of the total value and the total payment, minus the total user cost incurred, is at least the optimal welfare.
This unfortunately does not give us any nontrivial bound on the welfare in equilibrium.

In fact, we will show in Section~\ref{sec:global} that the second-price auction with cost-adjusted bids has a PoA of $0$.
Indeed, to handle user costs we need fundamentally new ideas and techniques.

\section{Our Mechanisms}

We present two quasilinear truthful mechanisms with provable efficiency guarantees in the presence of user costs.
The first mechanism is anonymous, i.e., the outcome that any bidder receives is indepedent of that bidder's identity (or index).
The second mechanism is auction-oblivious, i.e., it works in the same way across all auctions, and may treat different bidders in different ways.

\subsection{Cost Multipliers}

We first introduce the key idea behind both of our mechanisms, the use of cost multipliers.
Instead of adjusting each bidder's bid by subtracting the user cost as is (as in the second-price auction with cost-adjusted bids discussed in the previous section), we apply a multiplier to the cost before performing the subtraction.
We illustrate the power of this idea in the single-bidder setting, which also serves as a warm-up for our subsequent discussion.

\paragraph{Cost multipliers in the single-bidder setting.}
Suppose $n = 1$.
For each auction $j$, let $v_j$, $c_j$, and $b_j$ be the only bidder's value, user cost, and bid in $j$, respectively. 
Let $S$ be the set of auctions where the bidder's value is no smaller than the respective cost, i.e., $S = \{j \mid v_j \ge c_j\}$.
Then, the optimal welfare is simply $\sum_{j \in S} (v_j - c_j)$.

Let $\alpha \ge 1$ be the unique number such that $\sum_{j \in S} v_j = \alpha \cdot \sum_{j \in S} c_j$.
Consider the following mechanism: in each auction $j$, the only bidder wins if the bid $b_j \ge \alpha \cdot c_j$, in which case the bidder pays the minimum bid to win, i.e., $p_j = \alpha \cdot c_j$; otherwise, no bidder wins and no payment happens.
Observe that the bidder's optimal bidding strategy in response to this mechanism is to set their bid multiplier to exactly $\alpha$, such that $b_j = \alpha \cdot v_j$ in each $j$.
This is because when the bid multiplier is $\alpha$, the bidder wins in all auctions in $S$ and no other auctions, so the total value collected by the bidder is $\sum_{j \in S} v_j$.
On the other hand, in each $j \in S$, the payment made by the bidder is $p_j = \alpha \cdot c_j$, so by the choice of $\alpha$, the total payment is $\sum_{j \in S} \alpha \cdot c_j = \sum_{j \in S} v_j$, which is the same as the total value.
In other words, when the bid multiplier is $\alpha$, the bidder's ROI constraint is binding.
This means there is no way to win in more auctions without violating the ROI constraint, because the marginal ROI ratio of winning in any other auction is strictly smaller than $1 / \alpha \le 1$.
So, the only equilibrium is fully efficient, and the PoA of this mechanism in the single-bidder setting is $1$.

\subsection{Mechanisms with Global Cost Multipliers}
\label{sec:global}
However, it turns out that it is not enough to use a {\em global} cost multiplier. In particular, mechanisms with global cost multipliers can get arbitrarily worse PoA.

\begin{theorem}
For any $\delta \in (0, 1/3)$, there exists an instance in which any mechanism with a global cost multiplier gets PoA $\leq 3\delta$.
\end{theorem}

\begin{proof}
    Consider the instance with $n = \lfloor 1/\delta \rfloor$ and $m = 2n$. Bidder $i$ only has non-zero values and costs in auction $2i-1$ and $2i$. The values and costs of bidder $i$ are specified as the follow table:
    \begin{center}
\begin{tabular}{ |c|c|c| } 
 \hline
  & Auction $2i-1$ & Auction $2i$ \\ 
 \hline
 Values ($v_{i,2i-1}$, $v_{i,2i}$)& $\delta$ & $1 + 1/\delta^i$ \\ 
 \hline
 Costs ($c_{i,2i-1}$, $c_{i,2i}$)& $1-\delta$ & $1/\delta^i$ \\ 
 \hline
\end{tabular}
\end{center}

It is easy to check that the optimal welfare can be achieved by letting bidder $i$ only wins auction $2i$, i.e.
\[
\text{OPT} = \sum_j \max\left\{0, \max_i (v_{i, j} - c_{i, j})\right\} = n.
\]

Notice that 
\[
\frac{v_{i,2i}}{c_{i,2i}} = 1 + \delta^i,
\]
and
\[
\frac{v_{i,2i-1} + v_{i,2i}}{c_{i,2i-1} + c_{i,2i}} = 1 + \frac{2\delta}{1 - \delta + \delta^{-i}} > 1 + \frac{2\delta}{2\delta^{-i}} = 1 + \delta^{i+1}.
\]

For any mechanism with a global cost multiplier $\alpha$, if $\alpha > 1 + \delta^i$, bidder $i$ will win nothing and contribute 0 to the welfare. And if $\alpha \leq 1 + \delta^{i+1}$, bidder $i$ will win both auction $2i-1$ and auction $2i$, and the overall welfare is $\delta + 1 + 1/\delta^i - (1 - \delta + 1/\delta^i) = 2\delta$.

Therefore, for each bidder $i \in [n]$, its contribution to welfare is at most $2\delta$ if $\alpha \not\in (1+\delta^{i+1}, 1+\delta^i]$, and at most 1 otherwise. Note that these intervals $(1+\delta^{i+1}, 1+\delta^i]$ are disjoint, and thus the overall welfare of a mechanism with any global multiplier $\alpha$ is at most
\[
2\delta \cdot (n -1) + 1 < (2\delta + 1/n) \cdot n \leq 3\delta n = 3\delta \cdot \text{OPT}. \qedhere
\]
\end{proof}

Such an impossibility result implies that it is necessary to incorporate additional information into the design of cost multipliers to be effective. In Section~\ref{sec:auction-dependent} and Section~\ref{sec:bidder-dependent}, we show the existence of effective auction-dependent cost multipliers and bidder-dependent cost multipliers, respectively. The constructions of these cost multipliers in the existence proofs rely on knowing additional auction-specific and/or bidder-specific value information, which may not be accessible in a real system. One could consider applying -driven techniques to construct cost multipliers in practice, which is outside of the scope of this paper.

\subsection{Auction-Dependent Cost Multipliers}
\label{sec:auction-dependent}

In this section, we show that there exist auction-dependent (and bidder-independent) cost multipliers with a PoA of $1/2$.

\paragraph{The mechanism.}
In each auction $j$, if $\max_i (v_{i, j} - c_{i, j}) < 0$, then no bidder wins (we let $\rw(j) = 0$ in such cases).
Otherwise, let $\rw(j) = \argmax_i (v_{i, j} - c_{i, j})$ (breaking ties arbitrarily), and let $\alpha_j = (v_{\rw(j), j} - c_{\rw(j), j}) / (2 c_{\rw(j), j})$.
Sort all bidders by the following score: $b_{i, j} - (1 + \alpha_j) \cdot c_{i, j}$.
Let $i_1(j)$ and $i_2(j)$ be the bidders with the highest and second highest scores respectively (again breaking ties arbitrarily).
$i_1(j)$ wins (i.e., $x_{i_1(j), j} = 1$) if their score is at least $0$, and pays the minimum bid to win, i.e.,
\[
    p_{i_1(j), j} = \max\{b_{i_2(j), j} - (1 + \alpha_j) \cdot c_{i_2(j), j}, 0\} + (1 + \alpha_j) \cdot c_{i_1(j), j}.
\]
If $i_1(j)$'s score is smaller than $0$, then no bidder wins.

One may check that the above mechanism is in fact quasilinear truthful.
There is one ambiguity in the mechanism: for an auction $j$, if $c_{\rw(j), j} = 0$, then the mechanism sets $\alpha_j = \infty$, which may introduce undefined behavior.
Although this corner case is unlikely to be a problem in practice, for consistency in theory, we use the following convention: if $\alpha_j = \infty$, then for any bidder $i$, $\alpha_j \cdot c_{i, j} = \infty$ if $c_{i, j} > 0$, and $\alpha_j \cdot c_{i, j} = \frac12 v_{\rw(j), j}$ if $c_{i, j} = 0$.
One may check that with this convention, our analysis correctly handles infinite multipliers.

\paragraph{Analysis of the mechanism.}
The efficiency of the mechanism is captured by the following theorem.

\begin{theorem}
    The mechanism with auction-dependent cost multipliers has a PoA of $1/2$.
\end{theorem}
\begin{proof}
    First observe that by Myerson's characterization~\citep{myerson1981optimal}, the mechanism is quasilinear truthful, and so each bidder $i$'s bid in each auction $j$ is at least their value, i.e., $b_{i, j} \ge v_{i, j}$.
    Using this fact, we will lower bound the difference between the total payment made in all auctions and the total cost incurred in all auctions in equilibrium.
    That is, we lower bound the following quantity:
    \[
        \sum_{i, j} (p_{i, j} - x_{i, j} \cdot c_{i, j}).
    \]
    Due to the ROI constraints, this is a lower bound of the welfare in equilibrium.
    Fix any equilibrium.
    We compare the contribution of each auction to the above difference against the contribution of the same auction to the optimal welfare.
    For each auction $j$, consider the following $3$ cases:
    \begin{itemize}
        \item $\max_i (v_{i, j} - c_{i, j}) < 0$ and no bidder wins.
        In this case, the payment made and the cost incurred in $j$ are both $0$.
        On ther other hand, the contribution of $j$ to the optimal welfare is also $0$.
        So the contribution of $j$ to the difference is (trivially) at least $1/2$ of the contribution of $j$ to the optimal welfare.
        \item $\rw(j)$ wins in auction $j$, i.e., $i_1(j) = \rw(j)$.
        In this case, by the choice of $\alpha_j$, the payment that $\rw(j)$ makes is at least
        \[
            p_{\rw(j), j} \ge (1 + \alpha_j) \cdot c_{\rw(j), j} = c_{\rw(j), j} + \frac12 (v_{\rw(j), j} - c_{\rw(j), j}).
        \]
        The cost incurred in $j$ is $c_{\rw(j), j}$, so the contribution of $j$ to the difference is at least $\frac12 (v_{\rw(j), j} - c_{\rw(j), j})$.
        On the other hand, the contribution of $j$ to the optimal welfare is $v_{\rw(j), j} - c_{\rw(j), j}$, and the ratio between the contribution to the difference and that to the optimal welfare is at least $1/2$.
        \item A bidder other than the rightful winner wins.
        In this case, we know the second highest score is at least the score of the rightful winner $\rw(j)$.
        Since the rightful winner bids at least their value, and by the choice of $\alpha$, we have
        \begin{align*}
            b_{i_2(j), j} - (1 + \alpha_j) \cdot c_{i_2(j), j} & \ge b_{\rw(j), j} - (1 + \alpha_j) \cdot c_{\rw(j), j} \\
            & \ge v_{\rw(j), j} - (1 + \alpha_j) \cdot c_{\rw(j), j} \\
            & = \frac12 (v_{\rw(j), j} - c_{\rw(j), j}).
        \end{align*}
        Then the payment of the winner is at least
        \begin{align*}
            p_{i_1(j), j} & \ge b_{i_2(j), j} - (1 + \alpha_j) \cdot c_{i_2(j), j} + (1 + \alpha_j) \cdot c_{i_1(j), j} \\
            & \ge \frac12 (v_{\rw(j), j} - c_{\rw(j), j}) + (1 + \alpha_j) \cdot c_{i_1(j), j}.
        \end{align*}
        The cost incurred in $j$ is $c_{i_1(j), j}$, so the contribution of $j$ to the difference is at least
        \begin{align*}
            p_{i_1(j), j} - c_{i_1(j), j} & \ge \frac12 (v_{\rw(j), j} - c_{\rw(j), j}) + \alpha_j \\
            & \ge \frac12 (v_{\rw(j), j} - c_{\rw(j), j}).
        \end{align*}
        In other words, the contribution of $j$ to the difference is at least $1/2$ of the contribution of $j$ to the optimal welfare.
    \end{itemize}

    Summarizing the $3$ cases, we see that the contribution of each auction $j$ to the difference between the total payment and the total cost is at least $1/2$ of the contribution of the same auction to the optimal welfare.
    As a result, the difference between the total payment and the total cost is at least $1/2$ of the optimal welfare.
    By the ROI constraints of the bidders, the former quantity is a lower bound of the welfare in equilibrium, which gives a lower bound of $1/2$ on the PoA.
\end{proof}

\subsection{Bidder-Dependent Cost Multipliers}
\label{sec:bidder-dependent}

We now proceed to show the existence of bidder-dependent (and auction-independent) cost multipliers with a PoA of $1/4$.

\paragraph{The mechanism.}
For each bidder $i$, let $S_i$ be the set of auctions where $i$ is the rightful winner (as defined in the previous subsection), i.e., $S_i = \{j \mid \rw(j) = i\}$.
Note that by definition, for each $i$ and $j \in S_i$, $v_{i, j} - c_{i, j} \ge 0$.
For each bidder $i$, let $\alpha_i \ge 0$ be the unique number such that
\[
    \sum_{j \in S_i} v_{i, j} = (1 + 2 \alpha_i) \cdot \sum_{j \in S_i} c_{i, j}.
\]
In each auction $j$, first discard any bidder $i$ where $b_{i, j} - (1 + \alpha_i) \cdot c_{i, j} < 0$.
If all bidders are discarded, then no bidder wins in auction $j$.
Otherwise, let each bidder $i$'s score be the cost-adjusted bid $b_{i, j} - c_{i, j}$ and sort all remaining bidders by their scores.
The bidder among the remaining ones with the highest score, $i_1(j)$, wins, and pays the minimum bid to win.
That is, if $i_1(j)$ is the only remaining bidder, then $i_1(j)$ pays
\[
    p_{i_1(j), j} = (1 + \alpha_{i_1(j)}) \cdot c_{i_1(j), j}.
\]
Otherwise, let $i_2(j)$ be the bidder among the remaining ones with the second highest score.
$i_1(j)$ pays
\[
    p_{i_1(j), j} = \max\{(1 + \alpha_{i_1(j)}) \cdot c_{i_1(j), j}, b_{i_2(j), j} - c_{i_2(j), j} + c_{i_1(j), j}\}.
\]

Again, one can check the above mechanism is in fact quasilinear truthful.
We also note that this mechanism, unlike our anonymous mechanism, uses different quantities to prescreen and sort bidders.
In particular, the quantity used for prescreening is no larger than the score used for sorting.
This is necessary for the mechanism to provide the desired efficiency guarantee.

\paragraph{Analysis of the mechanism.}
The efficiency of the mechanism is captured by the following theorem.

\begin{theorem}
\label{thm:bidder-dependent}
    The mechanism with bidder-dependent cost multipliers has a PoA of $1/4$.
\end{theorem}

The proof of the above theorem is fairly involved.
We first introduce some useful notions and properties.
For each bidder $i$, let $i$'s {\em core auctions} $T_i$ be the subset of $S_i$ where $i$'s value is at least $(1 + \alpha_i)$ times $i$'s user cost, i.e.,
\[
    T_i = \{j \in S_i \mid v_{i, j} \ge (1 + \alpha_i) \cdot c_{i, j}\}.
\]
We will use the following fact regarding $T_i$:
\begin{lemma}
\label{lem:T_i}
    For each bidder $i$,
    \[
        \sum_{j \in T_i} (v_{i, j} - c_{i, j}) \ge \frac12 \sum_{j \in S_i} (v_{i, j} - c_{i, j}).
    \]
\end{lemma}
\begin{proof}
    For each $i$,
    \begin{align*}
        \sum_{j \in T_i} (v_{i, j} - c_{i, j}) & = \sum_{j \in S_i} (v_{i, j} - c_{i, j}) - \sum_{j \in S_i \setminus T_i} (v_{i, j} - c_{i, j}) \\
        & = 2 \alpha_i \cdot \sum_{j \in S_i} c_{i, j} - \sum_{j \in S_i \setminus T_i} (v_{i, j} - c_{i, j}) \tag{choice of $\alpha_i$} \\
        & \ge 2 \alpha_i \cdot \sum_{j \in S_i} c_{i, j} - \sum_{j \in S_i \setminus T_i} \alpha_i \cdot c_{i, j} \tag{definition of $T_i$} \\
        & \ge \alpha_i \cdot \sum_{j \in S_i} c_{i, j} \\
        & = \frac12 \sum_{j \in S_i} (v_{i, j} - c_{i, j}). \tag{choice of $\alpha_i$}
    \end{align*}
    This establishes the claim.
\end{proof}

Fix an equilibrium.
We partition all bidders into {\em aggressive} ones $B_\agg$ and {\em conservative} ones $B_\con$, based on how large or small a bidder's bid multiplier is.
We say a bidder $i$ is aggressive ($i \in B_\agg$), if $i$'s bid multiplier $\theta_i$ is at least $\alpha_i$; $i$ is conservative otherwise:
\[
    B_\agg = \{i \mid \theta_i \ge \alpha_i\}, \qquad B_\con = \{i \mid \theta_i < \alpha_i\}.
\]

The plan is to consider two quantities separately, which each lower bound the welfare in equilibrium.
The first quantity, say $A$, is the total value collected, minus the total cost imposed, by all conservative bidders in their respective core auctions.
Formally,
\[
    A = \sum_{i \in B_\con, j \in T_i} x_{i, j} \cdot (v_{i, j} - c_{i, j}).
\]
The second quantity, say $B$, is the total payment made by all bidders, minus the total cost imposed, in auctions (1) where an aggressive bidder is the rightful winner, or (2) which are among the core auctions of a conservative bidder and that bidder does not win.
Formally,
\[
    B = \sum_{i \in B_\agg, j \in S_i} \sum_{i'} (p_{i', j} - x_{i', j} \cdot c_{i', j}) + \sum_{i \in B_\con, j \in T_i} \sum_{i' \ne i} (p_{i', j} - x_{i', j} \cdot c_{i', j}).
\]
The following lemma states that each of the two quantities is a lower bound of the welfare in equilibrium.

\begin{lemma}
\label{lem:AB}
    \[
        \max\{A, B\} \le \sum_{i, j} x_{i, j} \cdot (v_{i, j} - c_{i, j}).
    \]
\end{lemma}
\begin{proof}
    Consider $A$ first.
    We first decompose the welfare into the contirbution of each bidder.
    \[
        \sum_{i, j} x_{i, j} \cdot (v_{i, j} - c_{i, j}) = \sum_i \sum_{j: x_{i, j} = 1} (v_{i, j} - c_{i, j}).
    \]
    Observe that each bidder $i$'s contribution is nonnegative.
    This is because for each $i$,
    \begin{align*}
        \sum_{j: x_{i, j} = 1} (v_{i, j} - c_{i, j}) & \ge \sum_{j: x_{i, j} = 1} (p_{i, j} - c_{i, j}) \tag{$i$'s ROI constraint} \\
        & \ge \sum_{j: x_{i, j} = 1} ((1 + \alpha_i) \cdot c_{i, j} - c_{i, j}) \tag{payment rule} \\
        & = \alpha_i \cdot \sum_{j: x_{i, j} = 1} c_{i, j} \ge 0. \tag{$\alpha_i \ge 0$}
    \end{align*}
    So we have
    \[
        \sum_{i, j} x_{i, j} \cdot (v_{i, j} - c_{i, j}) \ge \sum_{i \in B_\con} \sum_{j: x_{i, j} = 1} (v_{i, j} - c_{i, j}).
    \]

    For each conservative bidder $i \in B_\con$, further observe that the contribution of each auction $j$ to the right hand side of the above inequality is nonnegative.
    This is because in order for $i$ to win in an auction $j$, $i$'s bid must at least satisfy
    \[
        b_{i, j} \ge (1 + \alpha_i) \cdot c_{i, j}.
    \]
    On the other hand, since $i$ is conservative, we know
    \[
        b_{i, j} = \theta_i \cdot v_{i, j} \le \alpha_i \cdot v_{i, j}.
    \]
    So in any auction $j$, $x_{i, j} = 1$ only if $v_{i, j} \ge c_{i, j}$.
    Given this, we can further relax the above inequality and get
    \[
        \sum_{i, j} x_{i, j} \cdot (v_{i, j} - c_{i, j}) \ge \sum_{i \in B_\con} \sum_{j \in T_i: x_{i, j} = 1} (v_{i, j} - c_{i, j}) = A.
    \]

    Now consider $B$.
    We have
    \begin{align*}
        & \,\phantom{=}\ \sum_{i, j} x_{i, j} \cdot (v_{i, j} - c_{i, j}) \\
        & = \sum_i \sum_{j: x_{i, j} = 1} (v_{i, j} - c_{i, j}) \\
        & \ge \sum_i \sum_{j: x_{i, j} = 1} (p_{i, j} - c_{i, j}) \tag{each $i$'s ROI constraint} \\
        & \ge \sum_{i \in B_\agg, j \in S_i} \sum_{i'} (p_{i', j} - x_{i', j} \cdot c_{i', j}) \\
        & \,\phantom{=}\ + \sum_{i \in B_\con, j \in T_i} \sum_{i' \ne i} (p_{i', j} - x_{i', j} \cdot c_{i', j}) \\
        & = B,
    \end{align*}
    where the last inequality is because by the payment rule, for each $i$ and $j$,
    \[
        p_{i, j} - x_{i, j} \cdot c_{i, j} \ge x_{i, j} \cdot ((1 + \alpha_i) \cdot c_{i, j} - c_{i, j}) \ge 0.
    \]
    So discarding the contribution of some auctions and some bidders cannot make the sum larger.
    This finishes the proof.
\end{proof}

Now we are ready to prove Theorem~\ref{thm:bidder-dependent}.

\begin{proof}[Proof of Theorem~\ref{thm:bidder-dependent}]
    Given Lemma~\ref{lem:AB}, the welfare in equilibrium can be bounded in the following way:
    \[
        \sum_{i, j} x_{i, j} \cdot (v_{i, j} - c_{i, j}) \ge \max\{A, B\} \ge \frac12 (A + B).
    \]
    So we only need to compare $A + B$ against the optimal welfare.
    Below we argue the former is at least $1/2$ of the latter, which gives a lower bound of $1/4$ on the PoA.

    First observe that the optimal welfare can be written as
    \[
        \sum_j \max\{\max_i (v_{i, j} - c_{i, j}), 0\} = \sum_i \sum_{j \in S_i} (v_{i, j} - c_{i, j}).
    \]
    Below we consider each bidder $i$ separately, and show that the contribution of auctions in $S_i$ to $A + B$ is at least half of the contribution of the same auctions to the optimal welfare, i.e., $\sum_{j \in S_i} (v_{i, j} - c_{i, j})$.
    First consider each $i \in B_\agg$.
    In this case, each $j \in S_i$ has a winner.
    This is because the bid $b_{i, j}$ of $i$ in $j$ satisfies $b_{i, j} = \theta_i \cdot v_{i, j} \ge (1 + \alpha_i) \cdot v_{i, j} \ge (1 + \alpha_i) \cdot c_{i, j}$, so at least one bidder (i.e., $i$) passes the prescreening stage of the mechanism.
    Also, observe that each $j \in B_\agg$ contributes nothing to $A$.
    So we only need to argue that the total contribution of all auctions in $S_i$ to $B$ is at least $1/2$ of the contribution of the same auctions to the optimal welfare.
    In fact, if $i$ wins in auction $j$, then the payment $i$ makes in $j$ is at least $(1 + \alpha_i) \cdot c_{i, j}$, so the contribution to $B$ is at least $\alpha_i \cdot c_{i, j}$.
    Summing over $j \in S_i$ and by the choice of $\alpha_i$, the total contribution to $B$ (and therefore to $A + B$) is at least
    \[
        \sum_{j \in S_i} \alpha_i \cdot c_{i, j} = \frac12 \sum_{j \in S_i} (v_{i, j} - c_{i, j}),
    \]
    which is $1/2$ of the contribution to the optimal welfare.

    Now consider each $i \in B_\con$.
    For such a bidder $i$, only auctions in $T_i$ may contribute to $A + B$.
    Observe that each $j \in T_i$ has a winner.
    This is becasue the bid $b_{i, j}$ is at least $v_{i, j}$, which by the choice of $T_i$ satisfies $v_{i, j} \ge (1 + \alpha_i) \cdot c_{i, j}$.
    So $b_{i, j} \ge (1 + \alpha_i) \cdot c_{i, j}$, and at least one bidder (i.e., $i$) passes the prescreening stage of the mechanism.
    Below we argue that the contribution of each $j \in T_i$ to $A + B$ is at least the contribution of the same auction to the optimal welfare.
    If $i$ wins in $j$, then $j$ contributes to $A$, and the contribution is precisely $v_{i, j} - c_{i, j}$, which is the same as the contribution of $j$ to the optimal welfare.
    If $i$ does not win in $j$, then the score of $i_2(j)$ is at least the score of $i$, i.e.,
    \[
        b_{i_2(j), j} - c_{i_2(j), j} \ge b_{i, j} - c_{i, j} \ge v_{i, j} - c_{i, j}.
    \]
    This means the payment made by the winner, $i_1(j)$ must satisfy
    \[
        p_{i_1(j), j} \ge b_{i_2(j), j} - c_{i_2(j), j} + c_{i_1(j), j} \ge v_{i, j} - c_{i, j} + c_{i_1(j), j}.
    \]
    So the contribution of $j$ to $B$ is at least $v_{i, j} - c_{i, j}$, which is the contribution of $j$ to the optimal welfare.
    Summing over $j \in T_i$, we see that the total contribution is at least
    \[
        \sum_{j \in T_i} (v_{i, j} - c_{i, j}),
    \]
    which by Lemma~\ref{lem:T_i} is at least
    \[
        \frac12 \sum_{j \in S_i} (v_{i, j} - c_{i, j}),
    \]
    which is the contribution of auctions in $S_i$ to the optimal welfare.

    Now we can conclude the proof by observing
    \begin{align*}
        \sum_{i, j} x_{i, j} \cdot (v_{i, j} - c_{i, j}) & \ge \max\{A, B\} \tag{Lemma~\ref{lem:AB}} \\
        & \ge \frac12 (A + B) \\
        & \ge \frac14 \sum_i \sum_{j \in S_i} (v_{i, j} - c_{i, j}) \\
        & = \frac14 \sum_j \max\{\max_i (v_{i, j} - c_{i, j}), 0\}.
    \end{align*}
    This gives a lower bound of $1/4$ on the PoA of the mechanism.
\end{proof}

\section{Conclusion}

In this paper, we study the impact of user costs in autobidding auctions. We show that running vanilla VCG may result in poor performance. To improve the performance of VCG, we propose new variants with cost multipliers and provably show that they can guarantee constant approximations in terms of social welfare. 
For future works, as our paper examines one possibility to improve VCG by cost multipliers, it is interesting to explore other variants of VCG that may improve the performance. It is also very interesting to consider bridging the gap between theory and practice by designing practical mechanisms inspired by our theoretical results.


\bibliographystyle{apalike}
\bibliography{ref}

\begin{thebibliography}{}

\bibitem[Abrams and Schwarz, 2007]{AbramsS07}
Abrams, Z. and Schwarz, M. (2007).
\newblock Ad auction design and user experience.
\newblock In {\em Proceedings of the 3rd International Conference on Internet
  and Network Economics}, WINE'07, page 529–534, Berlin, Heidelberg.
  Springer-Verlag.

\bibitem[Aggarwal et~al., 2019]{aggarwal2019autobidding}
Aggarwal, G., Badanidiyuru, A., and Mehta, A. (2019).
\newblock Autobidding with constraints.
\newblock In {\em International Conference on Web and Internet Economics},
  pages 17--30. Springer.

\bibitem[Aggarwal et~al., 2006]{aggarwal2006truthful}
Aggarwal, G., Goel, A., and Motwani, R. (2006).
\newblock Truthful auctions for pricing search keywords.
\newblock In {\em Proceedings of the 7th ACM Conference on Electronic
  Commerce}, pages 1--7.

\bibitem[Babaioff et~al., 2021]{babaioff2020non}
Babaioff, M., Cole, R., Hartline, J., Immorlica, N., and Lucier, B. (2021).
\newblock Non-quasi-linear agents in quasi-linear mechanisms.
\newblock In {\em Proceedings of 12th Innovations in Theoretical Computer
  Science}.

\bibitem[Balseiro et~al., 2021a]{balseiro2021landscape}
Balseiro, S., Deng, Y., Mao, J., Mirrokni, V., and Zuo, S. (2021a).
\newblock The landscape of auto-bidding auctions: Value versus utility
  maximization.
\newblock In {\em Proceedings of the 22nd ACM Conference on Economics and
  Computation}.

\bibitem[Balseiro et~al., 2022a]{BalseiroDMMZ22}
Balseiro, S., Deng, Y., Mao, J., Mirrokni, V., and Zuo, S. (2022a).
\newblock Optimal mechanisms for value maximizers with budget constraints via
  target clipping.
\newblock In {\em Proceedings of the 23rd ACM Conference on Economics and
  Computation}.

\bibitem[Balseiro et~al., 2021b]{BalseiroDMMZ21}
Balseiro, S.~R., Deng, Y., Mao, J., Mirrokni, V.~S., and Zuo, S. (2021b).
\newblock Robust auction design in the auto-bidding world.
\newblock In {\em Annual Conference on Neural Information Processing Systems
  2021, NeurIPS 2021}, pages 17777--17788.

\bibitem[Balseiro et~al., 2022b]{balseiro2022dynamic}
Balseiro, S.~R., Mirrokni, V., Leme, R.~P., and Zuo, S. (2022b).
\newblock Dynamic double auctions: Toward first best.
\newblock {\em Operations Research}, 70(4):2299--2317.

\bibitem[Bateni et~al., 2014]{bateni2014multiplicative}
Bateni, M., Feldman, J., Mirrokni, V., and Wong, S. C.-w. (2014).
\newblock Multiplicative bidding in online advertising.
\newblock In {\em Proceedings of the fifteenth ACM conference on Economics and
  computation}, pages 715--732.

\bibitem[Benway and Lane, 1998]{Benway1998BannerBW}
Benway, J.~P. and Lane, D.~M. (1998).
\newblock Banner blindness: Web searchers often miss "obvious" links.

\bibitem[Clarke, 1971]{clarke1971multipart}
Clarke, E.~H. (1971).
\newblock Multipart pricing of public goods.
\newblock {\em Public choice}, 11(1):17--33.

\bibitem[Cramton et~al., 2006]{cramton2006combinatorial}
Cramton, P., Shoham, Y., and Steinberg, R. (2006).
\newblock {\em Combinatorial Auctions}.
\newblock MIT Press.

\bibitem[Deng et~al., 2022]{deng2022efficiency}
Deng, Y., Mao, J., Mirrokni, V., Zhang, H., and Zuo, S. (2022).
\newblock Efficiency of the first-price auction in the autobidding world.
\newblock {\em arXiv preprint arXiv:2208.10650}.

\bibitem[Deng et~al., 2021]{deng2021towards}
Deng, Y., Mao, J., Mirrokni, V., and Zuo, S. (2021).
\newblock Towards efficient auctions in an auto-bidding world.
\newblock In {\em Proceedings of The Web Conference 2021}.

\bibitem[Drèze and Hussherr, 2003]{doi:10.1002/dir.10063}
Drèze, X. and Hussherr, F.-X. (2003).
\newblock Internet advertising: Is anybody watching?
\newblock {\em Journal of Interactive Marketing}, 17(4):8--23.

\bibitem[Edelman et~al., 2007]{edelman2007internet}
Edelman, B., Ostrovsky, M., and Schwarz, M. (2007).
\newblock Internet advertising and the generalized second-price auction:
  Selling billions of dollars worth of keywords.
\newblock {\em American Economic Review}, 97(1):242--259.

\bibitem[Fadaei and Bichler, 2016]{fadaei2016truthfulness}
Fadaei, S. and Bichler, M. (2016).
\newblock Truthfulness and approximation with value-maximizing bidders.
\newblock In {\em International Symposium on Algorithmic Game Theory}, pages
  235--246. Springer.

\bibitem[Feldman et~al., 2007]{feldman2007budget}
Feldman, J., Muthukrishnan, S., Pal, M., and Stein, C. (2007).
\newblock Budget optimization in search-based advertising auctions.
\newblock In {\em Proceedings of the 8th ACM conference on Electronic
  commerce}, pages 40--49.

\bibitem[Goel et~al., 2019]{goel2019pareto}
Goel, G., Mirrokni, V., and Leme, R.~P. (2019).
\newblock Pareto efficient auctions with interest rates.
\newblock In {\em Proceedings of the AAAI Conference on Artificial
  Intelligence}, volume~33, pages 1989--1995.

\bibitem[Goel et~al., 2014]{goel2014clinching}
Goel, G., Mirrokni, V., and Paes~Leme, R. (2014).
\newblock Clinching auctions beyond hard budget constraints.
\newblock In {\em Proceedings of the fifteenth ACM conference on Economics and
  computation}, pages 167--184.

\bibitem[Golrezaei et~al., 2021]{golrezaei2021auctions}
Golrezaei, N., Lobel, I., and Paes~Leme, R. (2021).
\newblock Auction design for roi-constrained buyers.
\newblock {\em Proceedings of The Web Conference 2021}.

\bibitem[Groves, 1973]{groves1973incentives}
Groves, T. (1973).
\newblock Incentives in teams.
\newblock {\em Econometrica: Journal of the Econometric Society}, pages
  617--631.

\bibitem[hoan Cho and Cheon, 2004]{Cho04whydo}
hoan Cho, C. and Cheon, H.~J. (2004).
\newblock Why do people avoid advertising on the internet.
\newblock {\em Journal of Advertising}, pages 89--97.

\bibitem[Li et~al., 2013]{LiLL13}
Li, J., Liu, D., and Liu, S. (2013).
\newblock Optimal keyword auctions for optimal user experiences.
\newblock {\em Decis. Support Syst.}, 56(C):450–461.

\bibitem[Liaw et~al., 2022]{rFPApaper}
Liaw, C., Mehta, A., and Perlroth, A. (2022).
\newblock Efficiency of non-truthful auctions under auto-bidding.

\bibitem[Mehta, 2022]{Mehta22}
Mehta, A.~S. (2022).
\newblock Auction design in an auto-bidding setting: Randomization improves
  efficiency beyond vcg.
\newblock In {\em Proceedings of the ACM Web Conference 2022}, pages 173--181.

\bibitem[Mirrokni et~al., 2018]{mirrokni2018dynamic}
Mirrokni, V., Paes~Leme, R., Ren, R., and Zuo, S. (2018).
\newblock Dynamic mechanism design in the field.
\newblock In {\em Proceedings of the 2018 World Wide Web Conference}, pages
  1359--1368.

\bibitem[Mirrokni et~al., 2020]{mirrokni2020non}
Mirrokni, V., Paes~Leme, R., Tang, P., and Zuo, S. (2020).
\newblock Non-clairvoyant dynamic mechanism design.
\newblock {\em Econometrica}, 88(5):1939--1963.

\bibitem[Mitchell, 2022]{programmatic22}
Mitchell, E. (2022).
\newblock {Programmatic Advertising Explainer}.
\newblock
  \hyperref[Link]{https://www.insiderintelligence.com/content/programmatic-advertiser-explainer}.
\newblock Accessed: Oct 1, 2022.

\bibitem[Myerson, 1981]{myerson1981optimal}
Myerson, R.~B. (1981).
\newblock Optimal auction design.
\newblock {\em Mathematics of Operations Research}, 6(1):58--73.

\bibitem[Varian, 2007]{varian2007position}
Varian, H.~R. (2007).
\newblock Position auctions.
\newblock {\em International Journal of Industrial Organization},
  25(6):1163--1178.

\bibitem[Vickrey, 1961]{vickrey1961counterspeculation}
Vickrey, W. (1961).
\newblock Counterspeculation, auctions, and competitive sealed tenders.
\newblock {\em The Journal of finance}, 16(1):8--37.

\bibitem[Wilkens et~al., 2017]{wilkens2017gsp}
Wilkens, C.~A., Cavallo, R., and Niazadeh, R. (2017).
\newblock Gsp: the cinderella of mechanism design.
\newblock In {\em Proceedings of the 26th International Conference on World
  Wide Web}, pages 25--32.

\bibitem[Wilkens et~al., 2016]{wilkens2016mechanism}
Wilkens, C.~A., Cavallo, R., Niazadeh, R., and Taggart, S. (2016).
\newblock Mechanism design for value maximizers.
\newblock {\em arXiv preprint arXiv:1607.04362}.

\end{thebibliography}

\appendix

\end{document}